\documentclass[conference,a4paper]{IEEEtran}

\usepackage[utf8]{inputenc} 
\usepackage[T1]{fontenc}
\usepackage{url}
\usepackage{ifthen}
\usepackage{cite}
\usepackage[cmex10]{amsmath} 
\usepackage{amsthm,graphicx,epsfig,epstopdf,subfigure,amssymb,eufrak,mathrsfs,epsf,bm,algorithm,algorithmicx,algpseudocode}
\usepackage{color}

\newcommand{\snr}{\textrm{SNR}}

\DeclareMathAlphabet{\mathbbold}{U}{bbold}{m}{n}

\newtheorem{theorem}{Theorem}

\newtheorem{remark}{Remark}

\begin{document}
\title{On the Key Generation Rate of Physically Unclonable Functions} 


\author{%
  \IEEEauthorblockN{Yitao~Chen, Muryong~Kim, and Sriram~Vishwanath}
  \IEEEauthorblockA{ University of Texas at Austin\\
  					 Austin, TX 78701 USA\\
                     Email: \{yitaochen, muryong\}@utexas.edu, sriram@ece.utexas.edu}
}


\maketitle

\begin{abstract}
THIS PAPER IS ELIGIBLE FOR THE STUDENT PAPER AWARD.\\ 
In this paper, an algebraic binning based coding scheme and its associated achievable rate for key generation using physically unclonable functions (PUFs) is determined. This achievable rate is shown to be optimal under the generated-secret  (GS) model for PUFs.
Furthermore, a polar code based polynomial-time encoding and decoding scheme that achieves this rate is also presented.
\end{abstract}


\section{Introduction}
Physically unclonable functions (PUFs) form a promising innovative primitive that are increasingly gaining traction in the domains of authentication and secret key storage \cite{PUFintro1, PUFintro2, PUFintro3}. Instead of storing secrets in digital memory, PUFs derive a secret from the physical characteristics of an integrated circuit (IC) that form an inherent part of the device. Such a PUF can be obtained, as even though the mask and manufacturing process is relatively similar among  ICs built for a particular purpose, each IC is actually unique due to normal manufacturing variability. 

This unique behavior after manufacturing stems from a \emph{static randomness} due to technological dispersion. This static randomness was characterized by Pelgrom \cite{StaticRng}, and is known to follow a normal distribution. Unfortunately, PUF outputs are also subject to \emph{dynamic randomness} due to measurement noise, which is detrimental to the reliability of a PUF as a source for cryptographic elements. 

 In this paper, we understand the information theoretic limits of key generation using PUFs, given this static and dynamic randomness in the system.  As discussed in  \cite{PUFintro1, PUFintro2}, one of the central use-cases for PUFs is secret key generation, where this key is subsequently utilized in a variety of cryptographic algorithms.  A higher key generation rate implies greater security guarantees for the overall system, and therefore, our focus is to understand its limits, and to characterize coding schemes that approach these limits.


\subsection{Related Work and Our Contributions}
There is already a considerable body of work on combining PUFs with error correction coding schemes  to obtain reliable keys or secrets \cite{PUFintro3}.  Conventionally, these have combined BCH/RS codes with PUFs. More recently, \cite{BChen} presents simulation results on the combination of a polar code with a PUF, setting the stage for such a combination to be understood analytically. In parallel work to this paper, the authors et al. \cite{WZ} uncover the connection between PUF key generation problem and Wyner-Ziv problem \cite{WynerZiv}. And they study a nested polar codes construction scheme based on \cite{nestedPC}. 

In this work, we present a PUF key generation scheme based on a previously well studied model called {\em generated-secret} (GS) model. In \cite{GS1}, the authors present the region of achievable secret-key vs. privacy-leakage (key vs. leakage) rates for the GS model. In this paper, we show that the optimal key generation rate is achievable using algebraic binning with linear codes, and uncover the relation between PUF key generation problem and Slepian-Wolf problem. Further, we present  encoding and decoding algorithms using polar codes that achieve the optimal rate. Finally, we present simulation results to showcase the performance of our scheme.

Compared to existing literature, we find that our scheme results in a relatively straightforward interpretation of the PUF key generation problem, and results in a  key generation rate that is optimal for the GS model for PUFs. This is further expanded on in later sections of the paper.

%
%



The remainder of this paper is organized as follows. Section~\ref{sec:sys} provides the system model for PUF and formally defines the problem. Section~\ref{sec:main} shows the algebraic binning method and polar code construction achieving the maximal key generation rate. Section~\ref{sec:comp} compares our method to the existing other methods. Section~\ref{sec:sim} presents the simulation result. Finally, Section~\ref{sec:con} concludes the paper and gives future directions.

\section{System Model} \label{sec:sys}
Upper case letters represent random variables and lower case letters their realizations. A superscript denotes a vector of variables, e.g., $X^n=X_1,X_2,\ldots,X_n$, and a subscript denotes the position of a variable in a vector. Calligraphic letters such as $\mathcal{X}$ denote sets, and set sizes are written as $\lVert \mathcal{X} \rVert$. $H_b(x)=-x\log x-(1-x)\log(1-x)$ is the binary entropy function. The $*$-operator is defined as $p*x=p(1-x)+(1-p)x$. The operator $\oplus$ represents the element-wise modulo-2 summation. A binary symmetric channel (BSC) with crossover probability $p$ is denoted by BSC($p$). $Q_b(X)$ represents a binary quantizer that quantizes $X>0$ to $1$ and $X<0$ to $0$.

\begin{figure}[!hbp] 
  \centering
    \includegraphics[width=0.45\textwidth]{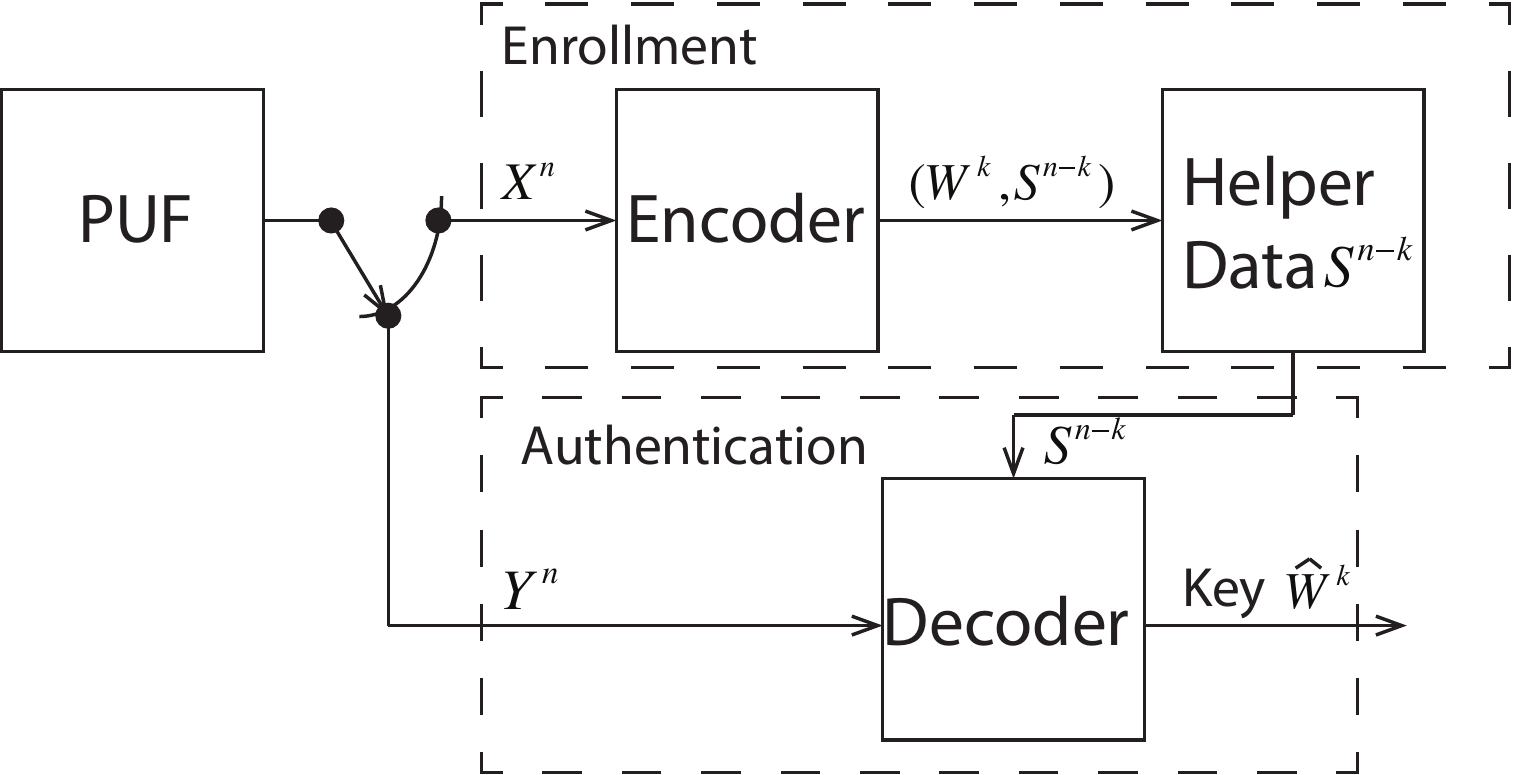}
   \caption{The GS model for PUF key generation system}
   \label{fig:1}
\end{figure}


A physically unclonable function (PUF) can be mathematically represented fairly simply as
\[
\tilde{Y} = \tilde{X} + \tilde{Z},
\]
where $\tilde{Y}$ is the PUF output, $\tilde{X}\sim \mathcal{N}(0,P)$ is the {\em static randomness}   and  $\tilde{Z} \sim \mathcal{N}(0,N)$ the {\em dynamic randomness}  independent of $X$.  As stated earlier in the introduction, $\tilde{X}$ is the desired ``signal", which is corrupted by the ``noise'' $\tilde{Z}$ when observed at the output of a PUF.

In most conventional systems today, the PUF output is quantized immediately after observation. Most models in literature assume the output  passes through a binary quantizer $Q_b(\cdot)$. It can be easily shown that $Q_b(\tilde{Y})$, $Q_b(\tilde{X})$ are distributed as $Bernoulli(0.5)$, and $Q_b(\tilde{Y})=Q_b(\tilde{X})\oplus Z'$, where $Z'$ is independent of $Q_b(\tilde{X})$, distributed as Bernoulli($p$), where $p$ is a function of $P$ and $N$. From now on, we use $Y$ to represent $Q_b(\tilde{Y})$, $X$ to represent $Q_b(\tilde{X})$ and $Z$ to represent $Z'$. Thus we have our PUF model with binary quantization as
\begin{flalign}\label{equ:1}
Y=X\oplus Z,
\end{flalign}
where $Y$ is the PUF output, $X \sim Bernoulli(0.5)$ and $Z \sim Bernoulli(p)$. Note that from the nature of PUF we do not have access to $X$. Indeed, in the real world applications, both $Y$ and $X$ are PUF outputs, which gives a different crossover probability $p*p$ in the model. But since it does not change the model (BSC), in the remaining part of this paper, we abuse the notations that we stick to the above model while both $Y$ and $X$ represent PUF outputs and $p$ is the parameter that measures the noise between two PUF outputs.

As mentioned earlier, we follow the {\em generated-secret} (GS) model for key generation of PUFs, as depicted in Figure \ref{fig:1}. For a given sequence $X^{n}$, our task is to design an encoder $\phi:\mathbb{F}_2^n :\rightarrow (\mathbb{F}_2^{n-k},\mathbb{F}_2^{k})$ that generates a {\em helper sequence} $S^{n-k}$  and a key $W^k$ and a decoder $\psi: \mathbb{F}_2^n \times \mathbb{F}_2^{n-k} :\rightarrow \mathbb{F}_2^k$ that authenticates the key. This is such that, for a particular PUF, the probability of successful authentication goes to $1$ as $n$ goes to infinity.  Let $(S^{n-k},W^k)=\phi(X^n)$ and $\hat{W}^k=\psi(Y^n,S^{n-k})$,
\begin{flalign*}
\Pr(\hat{W}^k \neq W^k) \rightarrow 0 \quad \text{as} \ n \rightarrow \infty.
\end{flalign*}

Define the key generation rate as $R=k/n$, we desire to determine the maximal key generation rate
\begin{flalign}
\max_{\phi,\psi} \ R \quad \text{s.t.}\  \Pr(\hat{W}^k \neq W^k) \rightarrow 0. 
\end{flalign}

\section{PUF System Design Using Polar Codes} \label{sec:main}
\begin{figure}[!hbp] 
  \centering
    \includegraphics[width=0.45\textwidth]{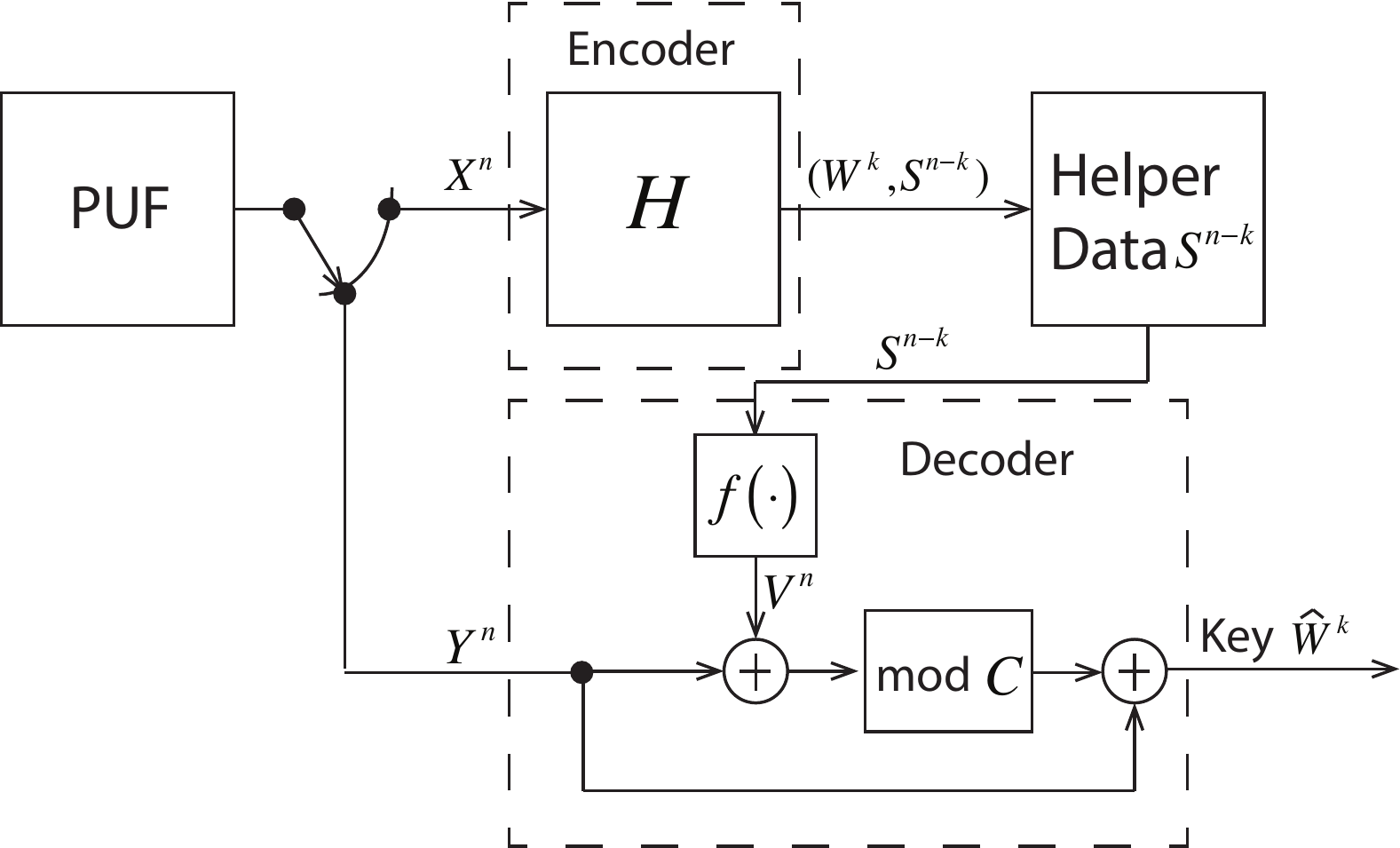}
   \caption{The GS model for PUF key generation system with algebraic binning}
   \label{fig:2}
\end{figure}

In this section, we first state our main theorem and show the optimal key generation rate is achievable with algebraic binning using linear codes. Note that, the result can also be obtained by random binning, but algebraic binning  offers greater insights for PUF key generation system design. Therefore, we choose to use an algebraic binning framework going forward.

\subsection{Algebraic Binning Using Linear Codes}

\begin{theorem}
Given a PUF and an associated key generation rate $R < I(X;Y)$, there exists a linear code such that
\begin{flalign*}
\Pr(\hat{W}^k \neq W^k) \rightarrow 0 \quad \text{as} \ n \rightarrow \infty.
\end{flalign*}
\end{theorem}

\begin{proof}
The basic idea is to generate the bins as the cosets of a ``good" parity-check code. Let an $(n,k)$ binary parity check code specified by the $(n-k) \times n$ binary parity-check matrix $H$. The code $\mathcal{C}=\{c^n\}$ contains all $n$-length binary vectors $c^n$ whose syndrome $s^{n-k}\triangleq H c^n$ is equal to zero, where here multiplication and addition are modulo 2. Assuming that all rows of $H$ are linearly independent, there are $2^k$ codewords in $\mathcal{C}$, so the code rate is $(\log |\mathcal{C}|)/n = k /n$. Given some general syndrome $s^{n-k} \in \{0,1\}^{n-k}$, the set of all $n$-length vectors $x^n$ satisfying $Hx^n=s^{n-k}$ is called a coset $\mathcal{C}_s$. Define a decoding function $f(s^{n-k})$, where $f:\{0,1\}^{n-k} \rightarrow \{0,1\}^n$, is equal to the vector $v^n \in \mathcal{C}_s$ with the minimum Hamming weight, where ties are broken evenly. It follows from linearity that the coset is a shift of the code $\mathcal{C}$ by the vector $v^n$, i.e.,
\begin{flalign*}
\mathcal{C}_s \triangleq \{x^n: Hx^n=s^{n-k}\}=\{c^n \oplus v^n : c^n \in \mathcal{C}\} \triangleq \mathcal{C}^v,
\end{flalign*}
where the $n$-length vector $v^n=f(s^{n-k})$ is  the {coset leader}.

Decoding of this parity-check code amounts to quantizing $y^n$ to the nearest vector in $\mathcal{C}$ with respect to the Hamming distance. This vector, $\hat{x}^n$, can be computed by syndrome decoding using the function $f$
\begin{flalign}\label{equ:altF}
\hat{x}^n = y^n \oplus \hat{z}^n, \quad \hat{z}^n=f(Hy^n)=y^n \  \text{mod} \ \mathcal{C}.
\end{flalign}
We may view the decoder above as a partition of $\{0,1\}^n$ to $2^{k}$ decision cells of size $2^{n-k}$ each, which are all shifted versions of the basic ``Voronoi'' set
\begin{flalign*}
\{z^n: z^n \oplus f(Hz^n)=0\} \triangleq \Omega_0.
\end{flalign*}
Each of the $2^{n-k}$ members of $\Omega_0$ is a coset leader for a different coset. The enrollment and authentication procedures with this algebraic binning are summarized in Algorithm~\ref{alg:1}. And the corresponding system model is shown in Figure~\ref{fig:2}.

Since our PUF model $Y=X \oplus Z$ is a BSC with crossover probability $p$. We are interested in ``good" parity check codes over BSC$(p)$ that are capacity achieving, i.e., they have a rate $R$ arbitrarily close to  $1-H(p)$ for $n$ large enough. Note that $H(X|Y)=H(p)$ and $H(X)=1$, together with Algorithm~\ref{alg:1} will grant us the desired result. 




\begin{algorithm}                               
\caption{Algebraic Binning}          
\label{alg:1}                                    
  \begin{algorithmic} 
 	  \Procedure{Enrollment}{$x^n$}
	  	\State Store the syndrome $s^{n-k}= Hx^n$ as helper data.
	  \EndProcedure
	  \Procedure{Authentication}{$y^n, s^{n-k}$}
	  	\State Find the coset leader $v^n=f(s^{n-k})$.
		\State Find the error between $y^n$ and the coset $\mathcal{C}_s$ 
		\State $\hat{z}^n=(v^n \oplus y^n) \ \text{mod} \ \mathcal{C}$.
		\State Reconstruct the key $\hat{w}^k = y^n \oplus \hat{z}^n \oplus v^n$. 
	  \EndProcedure
  \end{algorithmic}
\end{algorithm}


First, we show for any enrollment sequence $x^n$ and the corresponding PUF output $y^n$, the block error probability is vanishingly small. Note that the decoding computation in Algorithm~\ref{alg:1} is unique, so unlike in random binning we never have ambiguous decoding. Hence, noting from the PUF model $z^n=x^n \oplus y^n$ and from (\ref{equ:altF}) that $\hat{z}^n=f(H(x^n\oplus y^n))=f(Hz^n)$, a decoding error event amounts to $\{\hat{w}^k \neq w^k\} \leftrightarrow \{\hat{x}^n \neq x^n\} \leftrightarrow \{\hat{z}^n \neq z^n\}$ so the probability of decoding error is 
\begin{flalign*}
\Pr\{\hat{W}^k\neq W^k\}=\Pr\{\hat{X}^n\neq X^n\}=\Pr\{f(HZ^n) \neq Z^n\}
\end{flalign*}
which by good BSC$(p)$ code is smaller than $\epsilon$. 

Next, we show the optimal rate is $I(X;Y)$. Because the total number of typical sequences are $2^{nH(X)}$, maximizing the key generation rate $R$ is equivalent to minimize the number of bins (cosets) 
\begin{flalign}\label{equ:rate}
\max R = \max \ \frac{k}{n} = \max \ 1-\frac{n-k}{n}. 
\end{flalign}

Here we need the following Slepian-Wolf bound for distributed source coding. 
\begin{theorem} [Slepian-Wolf \cite{ThomasCoverBook}] 
For the distributed source coding problem for the source $(X,Y)$ drawn i.i.d.$\sim p(x,y)$, the achievable rate region is given by 
\begin{flalign*}
R_1 &\ge H(X|Y), \\
R_2 &\ge H(Y|X), \\
R_1+R_2 &\ge H(X,Y).
\end{flalign*}
\end{theorem}
To establish the connection between GS model and Slepian-Wolf problem, we see the two PUF outputs $X$ and $Y$ are the correlated source for Slepian-Wolf problem, and the number of bins in GS model is equivalent to the rate of the first source in Slepian-Wolf problem
\begin{flalign}\label{equ:SW}
\frac{n-k}{n}=R_1 \ge H(X|Y).
\end{flalign}
Combing Equ.(\ref{equ:rate}) and (\ref{equ:SW}), we have the optimal rate as $1-H(X|Y)=H(X)-H(X|Y)=I(X;Y)$.
The optimality is guaranteed by the Slepian-Wolf bound.
\end{proof}
\begin{remark}
The proof shows the key generation rate $I(X;Y)$ is achievable with a ``good'' coset partition, in a sense that each coset is a ``good'' parity check code over BSC($p$). It is a general statement, as long as one can find the ``good'' coset partition with each coset a ``good'' parity check code for some channel, the key generation rate $I(X;Y)$ is achievable for that channel.
\end{remark}


\subsection{Polar Codes for PUFs}

Polar codes are popular  linear block codes, introduced by Arikan in \cite{Arikan}. A binary polar code can be specified by $(N,K,\mathcal{F}, u^{\mathcal{F}})$, where $N=2^n$ is the block length, $K$ is the number of information bits encoded per codeword, $\mathcal{F}$ is the set of indices of the $N-K$ frozen bits and $u^{\mathcal{F}}$ is a vector of frozen bits, which is known to both encoder and decoder.

\subsubsection{Encoding of Polar Codes}
For an $(N,K,\mathcal{F}, u^{\mathcal{F}})$ polar code, the encoding operation for a message vector $u^N$, is performed using a generator matrix,
\begin{flalign*}
G_N=B_N G_2^{\otimes \log N},
\end{flalign*}
where $B_N$ is a {\em bit-reversal} permutation matrix, $G_2=\begin{bmatrix} 1 & 0 \\ 1 & 1 \end{bmatrix}$ and $\otimes$ denotes the Kronecker product. 

Given a message vector $u^N$, the codewords are generated as 
\begin{flalign*}
x^N=u^{\mathcal{F}^c} (G_N)_{\mathcal{F}^c} \oplus u^{\mathcal{F}} (G_N)_{\mathcal{F}}, 
\end{flalign*}
where $\mathcal{F}^c \triangleq \{1,2,\ldots, N\} \backslash \mathcal{F}$ corresponds to the information bits indices. So $u^{\mathcal{F}^c}$ are the information bits and $u^{\mathcal{F}}$ are the frozen bits. 

\subsubsection{Decoding of Polar Codes}
Polar codes achieve the channel capacity asymptotically in code length, when decoding is done using the successive-cancellation (SC) decoding algorithm. The SC decoder observes $(y^N, u^{\mathcal{F}})$ and generates an estimate $\hat{u}^N$ of $u^N$. The $i$th bit of the estimate $\hat{u}^N$ depends on the channel output $y^N$ and the previous bit decisions $\hat{u}_1, \hat{u}_2, \ldots, \hat{u}_{i-1}$, denoted by $\hat{u}^{i-1}$. It uses the following decision rules,
\begin{flalign*}
\hat{u}_i=\left\{ \begin{array}{ll}
 u_i & \textrm{if $i \in \mathcal{F}$}\\
 0 & \textrm{if $i \in \mathcal{F}^c$ and $L^{(i)}(y^N, \hat{u}^{i-1}) \ge 1$}\\
 l & \textrm{if $i \in \mathcal{F}^c$ and $L^{(i)}(y^N, \hat{u}^{i-1}) < 1$}
 \end{array} \right.
\end{flalign*}
where $L^{(i)}_N(y^N,\hat{u}^{i-1})=\mathbb{P}(y^N,\hat{u}^{i-1}|0)/\mathbb{P}(y^N,\hat{u}^{i-1}|1)$ is the $i$th likelihood ratio (LR) at length $N$. We omit further details in SC decoding for limited space, readers can get the full knowledge of SC decoding in \cite{Arikan}.

\subsection{Applying Polar Codes to PUFs} \label{sec:pcc}
Given the PUF model as a BSC($p$), block length $N$ and rate $K/N$, we have the polar code with parameters $(N,K,\mathcal{F})$. And the algebraic binning with polar code is shown in Algorithm~\ref{alg:2}.
\begin{algorithm}                               
\caption{Algebraic Binning with Polar Code Construction}          
\label{alg:2}                                    
  \begin{algorithmic} 
 	  \Procedure{Enrollment}{$x^N, (N,K,\mathcal{F})$}
	  	\State Store the syndrome $s^{N-K} = ((G_N^T)^{-1}x^N)_{\mathcal{F}}$ as helper data.
	  \EndProcedure
	  \Procedure{Authentication}{$y^N, s^{N-K}, (N,K,\mathcal{F})$}
		\State Reconstruct the key $\hat{w}^K = SCdec(y^n, s^{N-K})$. 
	  \EndProcedure
  \end{algorithmic}
\end{algorithm}


\begin{theorem}
For PUF, every key generation rate $R < I(X;Y)$, there exist a polar encoder and decoder, such that
\begin{flalign}\label{equ:err}
\sum_{s^{N-K} \in \mathcal{X}^{N-K}} \frac{1}{2^{N-K}} \sum_{w^K \in \mathcal{X}^K} \frac{1}{2^K}\Pr(\hat{w}^K \neq w^K) = O(N^{-\frac{1}{4}}).
\end{flalign}
\end{theorem}
\begin{proof}
As introduced in \cite{Arikan}, polar code can be represented as 
\begin{flalign*}
x^N&=u^{\mathcal{F}^c} (G_N)_{\mathcal{F}^c} \oplus u^{\mathcal{F}} (G_N)_{\mathcal{F}} \\
&=w^K (G_N)_{\mathcal{F}^c} \oplus s^{N-K} (G_N)_{\mathcal{F}}.
\end{flalign*}
By inversion of $G_N$, we have the syndrome and the key as
\begin{flalign*}
s^{N-K} &= ((G_N^T)^{-1}x^N)_{\mathcal{F}}, \\
w^K &= ((G_N^T)^{-1}x^N)_{\mathcal{F}^{c}}.
\end{flalign*}
Now for each PUF observation $x^N$, we treat it as a codeword of polar code with parameters $(N,K,\mathcal{F}, s^{N-K})$. So we use a set of polar code $\mathcal{C}=\{C(N,K,\mathcal{F},s^{N-K}):s^{N-K}\in \mathbb{F}_2^{N-K}\}$. Because $G_N$ has full rank, for any $x^N \in \mathbb{F}_2^{N}$, $x^N \in \mathcal{C}$, and $C \subseteq \mathbb{F}_2^{N}$. So $\mathcal{C}=\mathbb{F}_2^{N}$. We proved that $C$ is a coset partition of $\mathbb{F}_2^{N}$ and each coset code $C(N,K,\mathcal{F},s^{N-K})$ with coset leader $s^{N-K}$ is a channel code for the channel. According to the Theorem 3 in \cite{Arikan}, we have for rate $R < I(X;Y)$, the block error probability for polar coding under successive cancellation decoding satisfies 
\begin{flalign*}
\sum_{s^{N-K} \in \mathbb{F}_2^{N-K}} \frac{1}{2^{N-K}} \sum_{w^K \in \mathbb{F}_2^K} \frac{1}{2^K}\Pr(\hat{w}^K \neq w^K) = O(N^{-\frac{1}{4}}).
\end{flalign*}
Although polar codes cannot guarantee each coset code is a good channel code such that 
\begin{flalign*}
\Pr(\hat{W}^K \neq W^K)=o(1),
\end{flalign*}
on average, we obtain a good coset partition as required by Theorem 1.
\end{proof}


\subsection{Achievable Scheme for Unquantized PUFs: The Gaussian Case}

As mentioned earlier, a vast majority of PUF outputs are quantized to a binary alphabet right after generation. However, for the unquantized PUF model $\tilde{Y} = \tilde{X} + \tilde{Z}$ (again, we abuse the notations such that $\tilde{Y},\tilde{X}$ are outputs of the PUF), we use a lattice based coding scheme as below.

\emph{Definitions}: Lattice $\Lambda$ is a discrete subgroup of $\mathbb{R}^n$. Quantization with respect to $\Lambda$ is $Q_{\Lambda}(\tilde{x}^n)=\arg \min_{\lambda \in \Lambda} \|\tilde{x}^n-\lambda\|$. Fundamental Voronoi region of $\Lambda$ is $\mathcal{V}(\Lambda)=\{\tilde{x}^n: Q_{\Lambda}(\tilde{x}^n)=\mathbf{0}\}$. Volume of the Voronoi region of $\Lambda$ is $V(\Lambda)=\int_{\mathcal{V}(\Lambda)} d\tilde{x}^n$. Normalized second moment of $\Lambda$ is $G(\Lambda)=\frac{\sigma^2(\Lambda)}{V(\Lambda)^{2/n}}$ where $\sigma^2(\Lambda)=\frac{1}{nV(\Lambda)}\int_{\mathcal{V}(\Lambda)} \|\tilde{x}^n\|^2 d\tilde{x}^n$. A pair of Lattices $(\Lambda,\Lambda_0)$ are said to be nested if $\Lambda\subseteq\Lambda_0$.

 We use nested lattices $\Lambda \subseteq \Lambda_0$ for coset partitioning and algebraic binning. The encoder block with input $\tilde{x}^n$ and output $\tilde{d}^n$ is implemented by lattice modulo operation
\[\tilde{d}^n=[\tilde{x}^n]\textrm{ mod }\Lambda_0=\tilde{x}^n-Q_{\Lambda_0}(\tilde{x}^n).\]
We use $\tilde{d}^n$ as a helper data. As in Figure~\ref{fig:2}, for the decoder block with input $\tilde{y}^n$, helper data $\tilde{d}^n$, and output $\tilde{v}^n$, we perform
\[\tilde{v}^n = Q_{\Lambda_0}([\tilde{y}^n-\tilde{d}^n]\textrm{ mod }\Lambda).\]
Since
\[\tilde{y}^n = \tilde{x}^n + \tilde{z}^n = \tilde{t}^n + \tilde{d}^n + \tilde{z}^n\]
where $\tilde{t}^n=Q_{\Lambda_0}(\tilde{x}^n)\in\Lambda_0$, the decoder output is
\[\tilde{v}^n=Q_{\Lambda_0}([\tilde{t}^n+\tilde{z}^n]\textrm{ mod }\Lambda).\]
If we use nested lattices satisfying $\tilde{z}^n\in\mathcal{V}(\Lambda_0)$ with high probability, it follows that
\[Q_{\Lambda_0}([\tilde{t}^n+\tilde{z}^n]\textrm{ mod }\Lambda)=Q_{\Lambda_0}(\tilde{t}^n)\]
with high probability. Since
\[Q_{\Lambda_0}(\tilde{t}^n)=\tilde{t}^n=Q_{\Lambda_0}(\tilde{x}^n),\]
it also means that
\[Q_{\Lambda_0}([\tilde{y}^n-\tilde{d}^n]\textrm{ mod }\Lambda)=Q_{\Lambda_0}(\tilde{x}^n)\]
with high probability. In other words, the helper data cancels the effect of noise $\tilde{z}^n$.

The lattice codebook is defined by the set $\Lambda_0\cap\mathcal{V}(\Lambda)$.
The code rate is given by $R=\frac{1}{n}\log\left(\frac{V(\Lambda)}{V(\Lambda_0)}\right)$ where $V(\cdot)$ is the volume of the fundamental Voronoi region of a lattice. We use nested lattices with parameters $\sigma^2(\Lambda)=P$, $G(\Lambda)=\frac{1}{2\pi e}$, and $V(\Lambda)=(2\pi e P)^{\frac{n}{2}}$. Nested lattices good for Gaussian channel coding \cite{ErezZamir04} can be used to achieve a rate up to $R=\frac{1}{2}\log\left(\frac{P}{N}\right)$ with vanishing error probability. In practice, polar lattices \cite{YanLingWu13,YanLiuLing14} can be used for polynomial-time processing.


\section{Comparisons with Existing Methods} \label{sec:comp}
There are several existing method proposed for the GS model. 

The authors et al. \cite{WZ} establish the connection between Wyner-Ziv problem and the GS model. They describe the key-leakage-storage region for GS model. However, for GS model, according to the definitions, storage rate and privacy rate are the same since $I(X^n;W)=H(W)-H(W|X^n)=H(W)$, where $W$ is a function of $X^n$ in GS model. It is also reflected in Theorem 1 in \cite{WZ} as $R_l$ and $R_w$ have the same bound. So the key-leakage-privacy region can be treated as key-leakage region or key-storage region. We describe the optimal point of the key-storage region by the algebraic binning argument. The authors et al. \cite{WZ} also show a polar code construction based on the nested polar code in \cite{nestedPC} to achieve the key-leakage-storage region, which give the maximual key generation rate as $1-H_b(q*p)$ for given PUF noise as a BSC($p$), where $q \in [0, 0.5]$ is a chosen parameter for the first step {\em vector quantization} (VQ) in the nested polar code. Since $H_b(q*p) \ge H_b(p)$, we have our optimal rate greater than their optimal rate $1-H_b(p) \ge 1- H_b(q*p)$, and the storage $H_b(p) \le H_b(q*p)$. Notice that the both equalities can be achieved if $q=0$, but at this point they lose the nested polar code construction. The reason for the degradation in their result is that there exists a gap between Wyner-Ziv problem having distortion (reflected as the first step VQ in the nested polar code construction) and the GS model requiring an exact recovery of the key. So the VQ step introducing the distortion is not necessary for GS model. In all, we offer a better rate with a simpler implementation.

The authors et al. \cite{LDPC} offer an LDPC based scheme for PUF. Their scheme does not optimize the key generation rate since the LDPC does not necessarily form a coset partition. 


\section{Simulation Results} \label{sec:sim}
We simulate the system in Figure \ref{fig:2} with the polar code construction in Section \ref{sec:pcc} with MATLAB. If we use PUFs in a field programmable gate array (FPGA) as the randomness source, we must satisfy a block error probability $P_B$ of at most $10^{-6}$ \cite{FPGA}. Consider a BSC($p$) with crossover probability $p=0.15$, which is a common value for SRAM PUFs.

First, we consider the block length $N=1024$ and we design polar code with rate $K/N=128/1024=0.125$ for the BSC($p$) channel. We evaluate the block error probability performance of this code with SC decoder and SC list (SCL) decoder with list size 8 respectively for a BSC with a range of crossover probability, as shown in Figure.~\ref{fig:test}. It shows the SCL decoder has better performance, and achieves a block error probability of $P_B=10^{-6}$ at a crossover probability $0.2$. For comparison, we achieve the key generation rate $0.125$ with crossover probability and block error probability $(0.2, 10^{-6})$, better than the crossover probability and block error probability tuple $(0.1819, 10^{-6})$ in \cite{WZ}. 
\begin{figure}[!hbp] 
  \centering
    \includegraphics[width=0.5\textwidth]{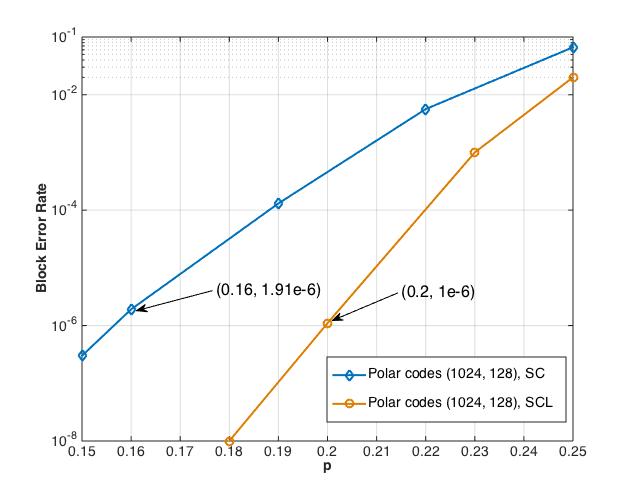}
   \caption{Block error probability over a BSC($p$) for polar codes with SC decoder and SCL decoder, respectively.}
   \label{fig:test}
\end{figure}
\section{Conclusion}\label{sec:con}
By algebraic binning, we show that ``good'' coset partition is needed to achieve the optimal key generation rate for PUFs. Thus we offer a principle in general for PUF key generation system design. And we design a polar code-based system for PUFs that achieve better key generation rate than existing methods. In future work, we will further study the ``good'' code for unquantized PUFs.

\section*{Acknowledgment}

This work was supported by the NSF and the ONR.




\end{document}